\documentclass[10pt,conference,romanappendices]{IEEEtran}
\usepackage{bm,cite}
\usepackage{amsmath}
\usepackage{amsthm}
\usepackage{amsbsy}
\usepackage{epsfig}
\usepackage{latexsym}
\usepackage{amssymb}
\usepackage{placeins}
\usepackage{ wasysym }

\newtheorem{theorem}{Theorem}

\newtheorem{lemma}{Lemma}
\newtheorem{corollary}{Corollary}
\DeclareMathAlphabet{\mathbit}{OML}{cmr}{bx}{it}
\DeclareMathAlphabet{\mathsf}{OT1}{cmss}{m}{n}
\DeclareMathAlphabet{\mathbsf}{OT1}{cmss}{bx}{it}

\newcommand{\Up}{{\text{up}}}

\newcommand{\Prob}{{\text{Pr}}}

\newcommand{\SINR}{{\text{SINR}}}

\newcommand{\E}{{\text{E}}}

\begin{document} 
\title{The Asymptotic Limits of Interference in Multicell Networks with Channel Aware Scheduling}
\author{Paul de Kerret and David Gesbert\\Mobile Communication Department, Eurecom\\
 2229 route des Cr\^etes, 06560 Sophia Antipolis, France\\\{dekerret,gesbert\}@eurecom.fr}

\maketitle

\begin{abstract} 
Interference is emerging as a fundamental bottleneck in many important wireless communication scenarios, including dense cellular networks and cognitive networks with spectrum sharing by multiple service providers.
Although multiple-antenna (MIMO) signal processing is known to offer useful degrees of freedom to cancel interference, extreme-value theoretic analysis recently showed that, even in the absence of MIMO processing, the scaling law of the capacity in the number of users for a multi-cell network with and without inter-cell interference was asymptotically identical provided a simple signal to noise and interference ratio (SINR) maximizing scheduler is exploited. This suggests that scheduling can help reduce inter-cell interference substantially, thus possibly limiting the need for multiple-antenna processing. However, the convergence limits of interference after scheduling in a multi-cell setting are not yet identified. In this paper\footnote{Discussions with Dr. Marios Kountouris are thankfully acknowledged. \\This work has been performed in the framework of the European research project ARTIST4G, which is partly funded by the European Union under its FP7 ICT Objective 1.1 - The Network of the Future.} we analyze such limits theoretically. We consider channel statistics under Rayleigh fading with equal path loss for all users or with unequal path loss. We uncover two surprisingly different behaviors for such systems. For the equal path loss case, we show that scheduling alone can cause the residual interference to converge to zero for large number of users. With unequal path loss however, the interference are shown to converge in average to a nonzero constant. Simulations back our findings. 
\end{abstract}  
\IEEEpeerreviewmaketitle
\section{Introduction}
Interference management has now appeared as one of the main challenges to reach the high data rates request in the future wireless networks. Cooperation and sharing of the users data between the base stations (BS) leads to the so-called multi-cell multiple-antenna (MIMO) network\cite{Somekh2009,Gesbert2010} which achieves high capacity. However, the requirement on the backhaul structure and the feedback are extremely high and have brought the need for more practical distributed approaches. Furthermore, multiuser diversity is well known to provide capacity gains via the scheduling of UEs with good channel gains \cite{Knopp1995} and appears as a promising tool to manage interference. For instance, in the single cell scenario, it was shown that thanks to multiuser diversity it was possible to achieve close to the performances of dirty paper coding with a simple random beamforming scheme, when the number of user equipments (UEs) becomes large \cite{Viswanath2002,Sharif2005}. This principle was extended to the multicell setting when communication is allowed between the BSs \cite{Tang2009}. In \cite{Choi2008}, the improvement brought by intercell scheduling was studied, and scheduling combined with a zero forcing (ZF) precoding scheme was discussed for the Wyner channel in \cite{Somekh2006}. The sum rate scaling was studied for particular types of information traffic when the number of antennas grows large in \cite{Al-Naffouri2008}, and asymptotically in terms of the number of transmitter-receiver pairs in \cite{Ebrahimi2007}. These works are among the many recent examples suggesting the beneficial impact of scheduling in interference-limited multiple-antenna networks. Beyond the assumption of some kind of multiple antenna processing, previous schemes must assume some fast exchange of channel state information between the cells, making it difficult to scale in certain practical scenarios with limited backhaul communications. 

In \cite{Gesbert2011}, the impact of the scheduler over the scaling law of capacity in many user networks was analyzed when all the BSs are assumed to transmit at maximal power, which allows for the distributed solutions and reduce largely the complexity for small costs. There, a \emph{simple distributed scheduler} (\emph{max-SINR} scheduler) has been shown to lead to the same scaling in terms of the number of UEs per cell as the \emph{no-interference upper bound}. This extreme-value theoretic result suggests that scheduling alone can signifcantly reduce the degradation brought by interference, thus confirming a general intuition in our community. Nevertheless it is not clear from existing studies whether scheduling alone can fully eliminate interference or just reduce it (even asymptotically) since this distinction is not visible from scaling law analysis. 

Our main contribution is to answer theoretically some of these questions and the main findings are now summarized. In the symmetric case, the average rate difference between the no-interference upper bound and the rate of the max-SINR scheduler is shown to converge to zero as a $O(log(log(n))/ log(n))$, where $n$ is the number of UEs per cell. The interference after scheduling also converge to zero. Interestingly, we prove that on the opposite the average interference power converge to a nonzero constant in the asymmetric case and that the average rate difference between the no-interference upper bound and the rate of the max-SINR scheduler does not tend to zero. 

\emph{Notations:} We write $g(n)=O(f(n))$ if $\exists N>0,\exists K>0, n\geq N\Rightarrow g(n)\leq K f(n)$. We use $g(n)=o(f(n))$ to denote that $\lim_{n\rightarrow\infty}g(n)/f(n)=0$ and we say that $g(n)$ is equivalent to $f(n)$ if $g(n)=f(n)+o(f(n))$, which we denote as $g(n)\sim f(n)$.
\section{System Model}
\subsection{Description of the Transmission Scheme}
We consider the downlink transmission in a multicell wireless network where in each cell one base station (BS) transmits to one user equipment (UE). The BSs and the UEs have each one antenna, and the UEs receive interference from $N$ neighboring cells (typically the first ring of interferers). We assume that $n$ UEs are located in each cell and apply single user decoding. The additive noise is assumed to be zero-mean white Gaussian with a variance of one. All the BSs transmit at their maximal power $P$, and the schedulers are applied at each BS separately, so that we study only one cell w.l.o.g. 

In the cell considered, the gain of the direct link to UE $k$ is denoted by $\gamma(k)G (k)\in \mathbb{R}^{+}$, where $\gamma(k)$ represents the path loss and $G(k)$ is a random variable modeling the short term fading. The gain of the link from the interfering BS~$j$ to UE~$k$ is denoted on the same pattern by $\gamma_{j}(k)G_{j}(k)\in \mathbb{R}^{+}$. Thus, the SINR at user $k$ is written as
\begin{equation}
\SINR(k)\triangleq\frac{\gamma(k)G(k)P}{1+\sum_{j=1}^{N} \gamma_{j}(k) G_{j}(k)P},\forall k\in[1,n].
\label{eq:system_model_1}
\end{equation}
We also define the short notations 
\begin{equation*}	
\alpha_k\triangleq P\gamma(k)G(k),\quad \beta_k\triangleq\sum_{j=1}^{N}\gamma_{j}(k)G_{j}(k)P,\quad \forall k\in[1,n].
\label{eq:system_model_2}
\end{equation*}  
We focus on distributed schedulers, i.e., which use only local CSI. The \emph{max-SINR} scheduler is an example of such schedulers. It maximizes the SINR of the UE, which reads as

\begin{equation}	
\Gamma\left(\alpha_n^{*},\beta_n^{*}\right))\triangleq\max_{k\in[1,n]}\left(\frac{\alpha_k}{1+\beta_k}\right).
\label{eq:system_model_3}
\end{equation}

We will compare this SINR with the \emph{no-interference} upper bound, obtained by considering a single cell setting with no interfering cell. The SINR after scheduling is then $\alpha^{\Up*}_n\triangleq\max_{k\in[1,n]} \alpha_k $. The rates are defined from the SINRs as 
\begin{equation}
\begin{aligned}
&R\left(\alpha^{*}_n,\beta^{*}_n\right)\triangleq\log_2\left(1+\Gamma\left(\alpha_n^{*},\beta_n^{*}\right)\right)\\
&R_{\Up}\left(\alpha^{\Up*}_n\right)\triangleq\log_2\left(1+\alpha^{\Up*}_n\right).
\end{aligned}
\label{eq:system_model_5}
\end{equation}
We will in fact focus on the average rates, and particularly on the difference between the average rates which we define as
\begin{equation}
\Delta_R(n)\triangleq\E\left[R_{\Up}\left(\alpha^{\Up*}_{n}\right)\right] -E\left[R\left(\alpha^{*}_n,\beta^{*}_{n}\right)\right].
\label{eq:system_model_6}
\end{equation}

\subsection{Channel Models}\label{se:system_model:channel_model}
We now recall the description of the two channel models considered \cite{Sharif2005,Gesbert2011} which we call the \emph{symmetric} and the \emph{asymmetric} model, depending on whether the UEs have equal path loss or not, respectively. For both cases, we model a cell as a disc of radius~$R$ instead of the hexagonal shape.

\subsubsection{The Symmetric Model}
In the symmetric model, we assume that all the UEs have the same average path loss. It is a general and interesting theoretical case, that we model in our cellular model by letting all the UEs be located at the same distance of the serving BS, i. e., on a circle of radius $R_{\text{sym}}$. The path loss is denoted as $\gamma(k)=\gamma$. 

\subsubsection{The Asymmetric Model}
In the asymmetric model, the UEs are distributed uniformly inside the disc of radius $R$. According to a generic path loss model, we have $\gamma(k)\!=\!\lambda d(k)^{-\varepsilon}$ and $\gamma_j(k)\!=\!\lambda d_j(k)^{-\varepsilon}$, with $\lambda$ a scaling factor, $\varepsilon$ the path loss exponent (usually $\varepsilon\!>2\!$), and $d(k)$ (resp. $d_j(k)$) the distance between the UE and the mother (resp. $j$-th interfering) BS. The exact shape of the cell has no impact asymptotically because the probability of scheduling a UE located at the distance $d\!>\!0$ of the BS vanishes as the number of UEs increases.

\section{Asymptotic Analysis in Symmetric Networks}
%
%
%
\begin{lemma}
In symmetric networks with a large number of UEs, the average rate $E\left[R_{\Up}\left(\alpha^{\Up*}_n\right)\right]$ is upper bounded as
\begin{equation*} 
\E\left[R_{\Up}\left(\alpha^{\Up*}_n\right)\right]\leq f(n)+\frac{\log_2(\log(n))}{\log(n)}+ o\left(\frac{\log(\log(n))}{\log(n)}\right)
\label{eq:symmetric_1}
\end{equation*}
with $f(n)\triangleq \log_2(\rho\log(n))$, and lower bounded as 
\begin{equation*} 
\E\left[R_{\Up}\left(\alpha^{\Up*}_n\right)\right]\geq f(n)-\frac{\log_2(\log(n))}{\log(n)}+ o\left(\frac{\log(\log(n))}{\log(n)}\right).
\label{eq:symmetric_2}
\end{equation*}
\label{lemma1}
\end{lemma}
\begin{proof}
We will use Theorem~A.2 from the Appendix of \cite{Sharif2005}, which in the case of a sequence of i.i.d. $\chi^2(2)$ random variables $x_i,i\in\{1,\ldots,n\}$ simplifies to:
\begin{equation}
\forall u\in \mathbb{R},\Prob\{w_n\leq u_n+u\}\rightarrow e^{-e^{-u+O\left(\tfrac{e^{-u}}{n}\right)}}
\label{eq:proof_sym_1}
\end{equation} 
where $w_n\triangleq \max_{i\in\{1,\ldots,n\}} x_i$ and $u_n\triangleq\log(n)$. We choose $u=\log(\log(n))$, which gives 
\begin{equation}
\Prob\left\{w_n \leq u_n+\log(\log(n))\right\}\sim 1-\frac{1}{\log(n)}
\label{eq:proof_sym_2}
\end{equation} 
and  $u=-\log(\log(n))$, to get
\begin{equation}
\Prob\left\{w_n \leq u_n-\log(\log(n))\right\}\sim\frac{1}{n}.
\label{eq:proof_sym_3}
\end{equation}
Considering $u=2\log(\log(n))$, it holds that 
\begin{equation}
\Prob\left\{w_n \leq\! u_n\!+2\log(\log(n))\right\}\!\sim\!1\!-\frac{1}{\log(n)^{2}}.
\label{eq:proof_sym_4}
\end{equation}
We have $\alpha^{\Up*}_n=\rho w_n$ with $\rho\triangleq P\gamma$, so that we write
{\small \begin{equation}
\begin{aligned}
&\E\left[R_{\Up}\left(\alpha^{\Up*}_n\right)\right]=\bigg(\Prob\left\{\alpha^{\Up*}_n\leq \rho\log(n)\!+\!\rho\log(\log(n))\right\}\\
&+\!\Prob\{\rho\log(n)\!+\!\rho\log(\log(n))\!<\!\alpha^{\Up*}_n\!\leq\! \rho\log(n)\!+2\rho\log(\log(n))\!\}\\
&+\!\Prob\{ \rho\log(n)\!+2\rho\log(\log(n))<\!\alpha^{\Up*}_n\}\bigg)
\!\log_2\!\left(1\!+\alpha^{\Up*}_n\right).
\end{aligned}
\label{eq:proof_sym_5}
\end{equation}}
which can be upper bounded using \eqref{eq:proof_sym_2} and \eqref{eq:proof_sym_4} as
{\small 
\begin{align}
&\E\left[R_{\Up}\left(\alpha^{\Up*}_n\right)\right] \nonumber\\ 
&\leq\left(1\!-\!\frac{1}{\log(n)}\!+\!o\left(\frac{1}{\log(n)}\right)\right)\log_2 \left(1\!+\!\rho\log(n)\!+\!\rho\log(\log(n))\right)\nonumber\\
&\!+\!\left(\frac{1}{\log(n)}\!+\!o\left(\frac{1}{\log(n)}\right)\right)\log_2\left(1\!+\!\rho\log(n)\!+\!2\rho\log(\log(n))\right)\nonumber\\
&\!+\!\left(\frac{1}{\log(n)^{2}}+o\left(\frac{1}{\log(n)^{2}}\right)\right)\log_2\left(1\!+\!\rho n\right)
\label{eq:proof_sym_5} 
\end{align} }
and after some simplifications yields
{ 
\begin{align}
&\E\left[R_{\Up}\left(\alpha^{\Up*}_n\right)\right]\leq\log_2\!\left(1\!+\!\rho\log(n)\right)\!+\!\log_2\!\left(\!1\!+\!\frac{\rho\log(\log(n))}{1\!+\!\rho\log(n)}\!\right)\!\!\nonumber\\
&-\!\!\frac{\log_2\!\left(1\!+\!\rho\log(n)\right)}{\log(n)}\!+\!\frac{\log_2\!\left(1\!+\!\rho\log(n)\right)}{\log(n)}\!+\!o\left(\frac{\log(\log(n))}{\log(n)}\right)\nonumber\\
&\sim\log_2\left(1+\rho\log(n)\right)+\frac{\log_2(\log(n))}{\log(n)}.
\label{eq:proof_sym_5}
\end{align}}
Using \eqref{eq:proof_sym_3}, a lower bound is derived similarly  as
{ 
\begin{align}
&\E\left[R_{\Up}\left(\alpha^{\Up*}_n\right)\right]\!=\!\bigg(\Prob\{\rho\log(n)\!-\!\rho\log(\log(n))\leq \!\alpha^{\Up*}_n\}\nonumber\\
&\!+\!\Prob\{\alpha^{\Up*}_n\!<\!\rho\log(n)\!-\!\rho\log(\log(n))\}\bigg)\log_2\left(1\!+\!\alpha^{\Up*}_n\right)\nonumber\\
&\geq\left(1\!-\!O\left(\frac{1}{n}\right)\right)\log_2\left(1\!+\!\rho\log(n)\!-\!\rho\log(\log(n))\right)\nonumber\\
&\sim\log_2\left(1+\rho\log(n)\right)+\frac{\log_2(\log(n))}{\log(n)}.
\label{eq:proof_sym_6} 
\end{align}}
\end{proof}

\begin{lemma}
In symmetric networks with the same path loss $\gamma$ from all BSs, $E\left[R\left(\alpha^{*}_n,\beta^{*}_n\right)\right]$ can be upper bounded as 
\begin{equation*} 
\E\left[R\left(\alpha^{*}_n,\beta_n^ {*}\right)\right]\leq\! f(n)-(N\!-\!1)\frac{\log_2(\log(n))}{\log(n)}\!+\! o\!\left(\!\frac{\log(\log(n))}{\log(n)}\!\right)
\label{eq:symmetric_1}
\end{equation*}
where $f(n)\triangleq \log_2(\rho\log(n))$, and lower bounded as 
\begin{equation*} 
\E\left[R\left(\alpha^{*}_n,\beta_n^ {*}\right)\right]\geq\! f(n)-(N\!+\!1)\frac{\log_2(\log(n))}{\log(n)}\!+\! o\!\left(\!\frac{\log(\log(n))}{\log(n)}\!\right)
\label{eq:symmetric_2}
\end{equation*}
\label{lemma2}
\end{lemma}

%

%
\begin{proof}
Only a sketch of the proof is given and the full proof can be found in \cite{dekerret2011jb}. We apply Lemma~$4$ and Corollary A.1 in \cite{Sharif2005} for the single BS transmission:

{\small
\begin{align}
&\Prob\{\Gamma\left(\alpha^{*}_n,\beta^{*}_n\right)\!\leq\!\rho\log(n)\!-\!\rho(N\!-\!1)\log(\log(n))\!\}\sim 1\!-\!\frac{1}{\log(n)}\nonumber\\ 
&\Prob\{\Gamma\left(\alpha^{*}_n,\beta^{*}_n\right)\!\leq\!\rho\log(n)\!-\!\rho(N-2)\log(\log(n))\!\}\sim 1\!-\!\frac{1}{(\log(n))^2}\nonumber\\ 
&\Prob\{\Gamma\left(\alpha^{*}_n,\beta^{*}_n\right)\!\leq\!\rho\log(n)\!-\!\rho(N+1)\log(\log(n))\!)\}
\sim \frac{1}{n} .
\label{eq:proof_sym_9} 
\end{align}
} 

The proof ends using the relations from \eqref{eq:proof_sym_9} to lower and upper bound the averate rate $\E\left[R\left(\alpha^{*}_n,\beta_n^ {*}\right)\right]$ as it has been done for the no-interference upper bound in the proof of Lemma~\ref{lemma1}.
\end{proof}

\begin{corollary}
The bounds of Lemma~\ref{lemma2} hold for any nonzero path loss of the interference.
\label{corollary1}
\end{corollary}
\begin{proof}
We will proceed by considering a lower and an upper bound for the path loss from the BSs to the UE. This is easily done by taking the maximum and the minimum of the distance between the interfering BSs. Then, for example with a lower bound for the path loss $\gamma_{\text{bd}}$, we write
\begin{equation*}
\SINR_{\text{bd}}(k)\triangleq \frac{\gamma}{\gamma_{\text{bd}}}\frac{G(k)}{\frac{1}{P\gamma_{\text{bd}}}+\sum_{j=1}^{N}G_{j}(k)}\triangleq
\theta_{\text{bd}} \SINR_{\text{bd,norm}}(k)
\end{equation*}
where we have defined $\SINR_{\text{bd,norm}}(k)$ as the SINR normalized with $\gamma_{\text{bd}}$, $\theta_{\text{bd}}\!\triangleq\!\gamma/\gamma_{\text{bd}}$, and $\rho_{\text{bd}}\!\triangleq \!P\gamma_{\text{bd}}$. The bounds in \eqref{eq:proof_sym_9} are linear in $\rho$ and are multiplied by $\theta_{\text{bd}}$ when inserted in the rates, so that they depend only on $\theta_{\text{bd}}\rho_{\text{bd}}$. We note that $\theta_{\text{bd}}\rho_{\text{bd}}\!=\!P\gamma\!=\!\rho$. As a consequence the bounds are independent of the path loss $\gamma_{\text{bd}}$, which ends the proof.
\end{proof}

\begin{theorem}
In symmetric networks, the rate difference $\Delta_R(n)$ vanishes as the number of UEs $n$ tends to infinity. To quantify the rate of convergence, we have the two following bounds:
\begin{equation}
\begin{aligned}
\Delta_R(n)&\leq (N\!+\!2)\frac{\log_2(\log(n))}{\log(n)}+ o\left(\!\frac{\log(\log(n))}{\log(n)}\!\right)\\
\Delta_R(n)&\geq (N\!-\!2)\frac{\log_2(\log(n))}{\log(n)}+ o\left(\!\frac{\log(\log(n))}{\log(n)}\!\right).
\end{aligned}
\label{eq:proof_sym_1}
\end{equation}
\label{theorem1}
\end{theorem}

\begin{proof}
The proof follows directly from the lemmas.
\end{proof}

\begin{corollary}
The interference $\beta^{*}_n$ tend almost surely to zero.
\label{corollary1}
\end{corollary}

\begin{proof}
We will show that assuming the interference not to tend almost surely (a.s.) to zero leads to a contradiction:
\begin{equation*} 
\exists\beta^*_{\infty}>0,K\in \mathbb{N}:n\geq K\rightarrow \Pr\{\beta^*\geq\beta^*_{\infty}\}=P_1>0.
\end{equation*} 
Then, for $n\geq K$,
\begin{equation}
\begin{aligned}
\Delta_R(n)&=\E\left[\log_2\left(\frac{\alpha^{\Up*}_n}{\Gamma\left(\alpha^{*}_n,\beta^{*}_n\right)}\right)\right]+o(1)\\ 
				&\geq \E\left[\log_2\left(1+\beta^{*}_n\right)\right]+o(1)\\
				&\geq P_{1}\log_2\left(1+\beta^*_{\infty}\right)+o(1)>0
\end{aligned}
\label{eq:proof_sym_12} 
\end{equation} 
where the first inequality was obtained because $\alpha^{*}_n\leq \alpha^{\Up*}_n$ and the second inequality holds because the function in the expectation is always non negative. The last line is in contradiction with Theorem~\ref{theorem1}, which ends the proof.   
\end{proof}

\section{Asymptotic Analysis in Asymmetric Network}

We start by showing a useful lemma, which reads as follows.
\begin{lemma}
In asymmetric networks, the average rate for the no-interference upper bound is asymptotically equivalent to
\begin{equation}
\E\left[R_{\Up}\left(\alpha^{\Up*}_n\right)\right]\sim\frac{\varepsilon}{2}\log(n).
\label{eq:asymmetric_proof_3}
\end{equation}  
\label{lemma3}
\end{lemma}
\begin{proof}
Only a sketch of the proof is given and the full proof can be found in \cite{dekerret2011jb}. From Lemma~$5$ in \cite{Gesbert2011}, we get 
\begin{equation*}
\lim_{n\rightarrow \infty}\Prob\{\alpha^{\Up*}_n\leq \lambda \E[Y^{\frac{2}{\varepsilon}}]^{\frac{\varepsilon}{2}}n^{\frac{\varepsilon}{2}}t\}=e^{-t^{-\tfrac{2}{\varepsilon}}},\forall t>0.
\label{eq:asymmetric_proof_1}
\end{equation*}
Using $t=\log(n)$ and $t=1/\log(n)$, we obtain
\begin{equation}
\begin{aligned}
\lim_{n\rightarrow \infty} \Prob\{\alpha^{\Up*}_n&\leq \beta E[Y^{\frac{2}{\varepsilon}}]^{\frac{\varepsilon}{2}}n^{\frac{\varepsilon}{2}}\log(n)\}\!\sim\!1-\tfrac{1}{(\log(n))^{\tfrac{2}{\varepsilon}}}\\
\lim_{n\rightarrow \infty}\Prob\{\alpha^{\Up*}_n&\leq \beta E[Y^{\frac{2}{\varepsilon}}]^{\frac{\varepsilon}{2}}n^{\frac{\varepsilon}{2}}\log(n)^{-1}\}\!\sim \frac{1}{n^{\tfrac{2}{\varepsilon}}} .
\end{aligned}
\label{eq:asymmetric_proof_2}
\end{equation}
We then proceed as for Lemma~\ref{lemma1}, i.e., we compute a lower and an upper bound for the average rate using \eqref{eq:asymmetric_proof_2}.  
\end{proof}

\begin{theorem}
The average interference after max-SINR scheduling $\E[\beta^{*}_n]$ converge to a positive constant $\beta^{*}_{\infty}$.
\label{theorem2}
\end{theorem}

\begin{proof}
We will show the theorem by contradiction. Let assume that $\beta^{*}_n$ converge in average to zero. It implies that it converges a. s. to zero, which is written as
\begin{equation}
\forall \eta>0, \exists n_{\eta}>0, n>n_{\eta}\Rightarrow \beta^{*}_n<\eta,\text{ a. s. }
\label{eq:asymmetric_proof_4}
\end{equation}
Since $(\beta_k)_k$ is i.i.d. non negative, we have that $\Pr\{\beta_k<\eta\}=P_{\eta}$, with $P_{\eta}$ tending to zero if $\eta$ tends to zero. We now choose arbitrarily an $\eta>0$ and $n>n_{\eta}$, such that
\begin{equation}
\begin{aligned}
\Gamma\left(\alpha_n^{*},\beta_n^{*}\right)&\stackrel{a.s.}{=}\max_{k\in[1,n]}\left(\frac{\alpha_k}{1+\beta_k}\right),\text{ s.t. $\beta_k<\eta$} \\
&\leq\max_{k\in[1,n]} \alpha_k ,\text{ subject to $\beta_k<\eta$}.
\end{aligned}
\label{eq:asymmetric_proof_5}
\end{equation}
Inserting \eqref{eq:asymmetric_proof_5} in the average rate, we get
\begin{equation*}
\begin{aligned}
\E\left[R\left(\alpha_n^{*},\beta_n^{*}\right)\right]
&\leq\E\left[\log_2(1+\max_{k\in[1,n]}\alpha_k)\right],\text{  s.t. $\beta_k<\eta$}\\
&=\E\left[\log_2\left(\alpha^{\Up*}_{P_{\eta}n}\right)\right]+o(1)\\ 
&=\tfrac{\varepsilon}{4}\log_2(P_{\eta}n)\!+\!\tfrac{\varepsilon}{4}\log_2(P_{\eta}n)\!+\!\varepsilon(P_{\eta}n)\!+\!o(1)
\end{aligned}
\label{eq:asymmetric_proof_6}
\end{equation*}
where $\varepsilon(n)=o(\log(n)$. We define $g(n)\triangleq\tfrac{\varepsilon}{4}\log_2(n)+\varepsilon(n)$ which is increasing in $n$, for $n$ large enough. We then have 
\begin{equation*} 
\E\left[R\left(\alpha_n^{*},\beta_n^{*}\right)\right]\leq\tfrac{\varepsilon}{4}\log_2(P_{\eta}n)\!+\!g(P_{\eta}n) \leq\tfrac{\varepsilon}{4}\log_2(P_{\eta}n)\!+\!g(n). 
\label{eq:asymmetric_proof_6}
\end{equation*}
On the other side, we consider the lower bound obtained with a scheduler maximizing only the gain of the direct link:
\begin{equation*}
\begin{aligned}
\E\left[R\left(\alpha_n^{*},\beta_n^{*}\right)\right]
&\geq\E\left[\log_2\left(1+\frac{\max_{k\in[1,n]}\alpha_k}{1+\beta_k}\right)\right]\\
&=\E\left[\log_2\left(\alpha^{\Up*}_{n}\right)\right]-\E\left[\log(1+\beta_k)\right]+o(1)\\
&= \frac{\varepsilon}{4}\log_2(n)+g(n)-C_{\text{LB}}+o(1)
\end{aligned}
\label{eq:asymmetric_proof_7}
\end{equation*}
where we have defined the constant $C_{\text{LB}}\triangleq \E\left[\log_2(1+\beta_k)\right]$. The difference between the bounds is then equal to $\varepsilon/4\log_2(P_{\eta})+C_{\text{LB}}$. It holds for any $P_{\eta}$ and we can find $P_{\eta}$ so that the difference becomes negative. This is a contradiction, and we conclude that $\beta_n^*$ does not converge to zero. 

Moreover, the SINR is decreasing in $\beta_k$ and $(\beta_k)_k$ is statistically independent from $(\alpha_k)_k$. Thus, the max-SINR scheduler improves the distribution of $\beta_k$ in order to reduce its value and increasing the number of UEs can only lead to a reduction of $\E[\beta_n^{*}]$. $\E[\beta_n^{*}]$ is lower bounded by a positive number and non increasing, thus converges to a positive value.
\end{proof}

\begin{corollary}
The average rate difference $\Delta_R(n)$ does not tend to zero as $n\rightarrow \infty$.
\label{corollary1}
\end{corollary}
\begin{proof}
Using that $ \alpha^*_n<\alpha^{\Up*}_n$, we lower bound $\Delta_R(n)$ as
\begin{equation}
\begin{aligned}
\Delta_R(n)&=\E\left[\log_2\left(\alpha^{\Up*}_n\frac{1+\beta_n^*}{\alpha^{*}_n}\right)\right]+o(1)\\
\Delta_R(n)&\geq\E\left[\log_2\left(1+\beta_n^*\right)\right]+o(1).
\end{aligned}
\label{eq:proof_sym_12} 
\end{equation} 
The last term is the expectation of a non negative function such that it is equal to zero only if the function is equal to zero almost everywhere, which contradicts Theorem~\ref{theorem2}.
\end{proof}
\section{Simulations}
We simulate a multicell networks made of three cells with their first ring of interferers with full frequency reuse. We consider the propagation parameters of the LTE cellular network for the Hata urban scenario path loss model. Our parameters give an attenuation between a BS and a UE located at a distance of $d$ equal to $-114.5-37.19\log_{10}(d)$ dB with $d$ in km and the antennas gains taken into account. The transmit power is $P_{\text{dBm}}=40$ dBm per BS, the noise power $P_{\text{noise,dBm}}=-101$ dBm, and the radius of the cell $R=2$ kms. For the symmetric setting, the radius is $R_{\text{sym}}=1$ km. 

For comparison, a Joint Processing (JP) multicell transmission scheme with three cooperating BSs has been also used. It consists in applying a precoder removing the intercell interferences for the UEs selected distributively at the BSs. Waterfilling is then applied to the obtained diagonal multiuser channel, and we finally normalize the precoder to fulfill the per BS power constraint. For the sake of comparison with the JP multicell scheme, only the intracell interference have been removed in the "no-interference upper bound". We also compare the max-SINR scheduler with a scheduler maximizing the direct link gain, called the \emph{max-Gain} scheduler.

\begin{figure}[t] 
\centering
\includegraphics[width=1\columnwidth]{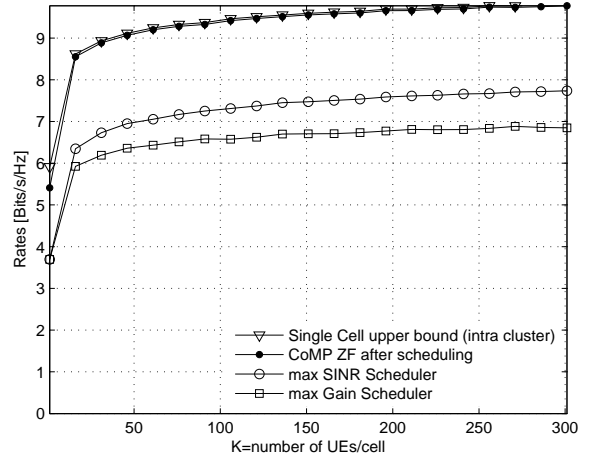}
\caption{Average rates in the symmetric model with $R_{\text{sym}}=1000$m as a function of the number of UEs per cell.}
\label{figure1}
\end{figure}

\begin{figure}[t] 
\centering
\includegraphics[width=1\columnwidth]{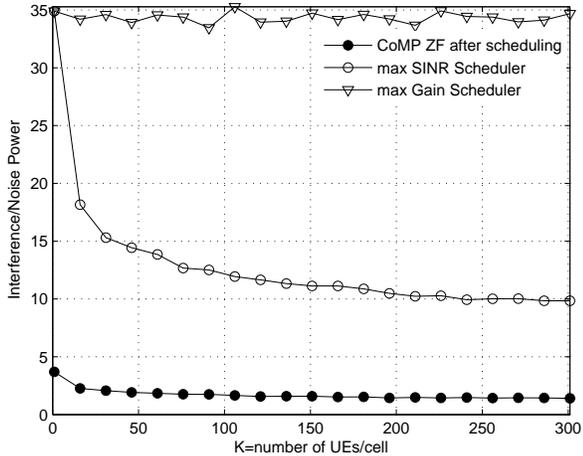}
\caption{Average interference power normalized by the noise variance in the symmetric case with $R_{\text{sym}}=1000$m in terms of the number of UEs per cell.}
\label{figure2}
\end{figure}

In Fig.~\ref{figure1}, the difference between the no-interference upper bound and the max-SINR scheduler is significant and decreases very slowly, while the rate of the max-Gain scheduler does not seem to converge to the single cell upper bound. Actually, the rate difference was not plotted due to space constraint, but the convergence to a positive constant is then very obvious. Note that the JP multicell scheme performs as the "no-interference upper bound". Moreover, in Fig.~\ref{figure2}, the remaining interference after scheduling decrease monotonically. Yet, the convergence to zero is very slow, and for realistic number of UEs per cell, the interference power remains much larger than the noise power.

\begin{figure}[t] 
\centering
\includegraphics[width=1\columnwidth]{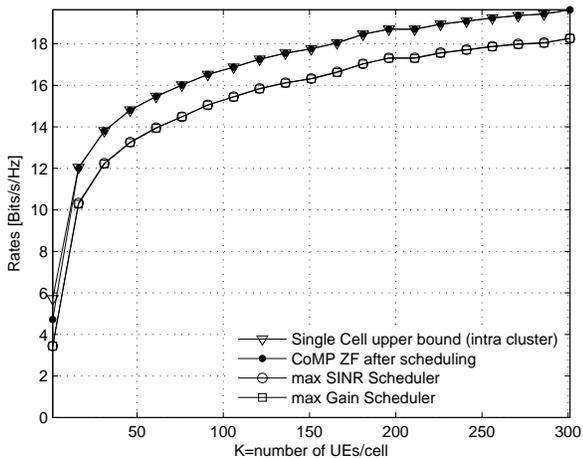}
\caption{Average rates in the asymmetric model as a function of the number of UEs per cell.}
\label{figure3}
\end{figure}

\begin{figure}[t] 
\centering
\includegraphics[width=1\columnwidth]{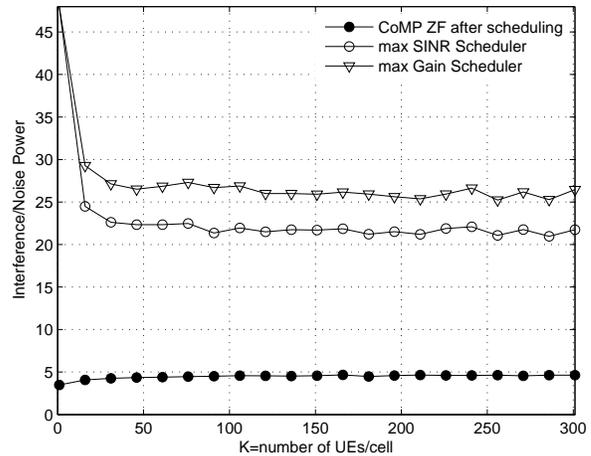}
\caption{Average interference power normalized by the noise variance in the asymmetric case in terms of the number of UEs per cell.}
\label{figure4}
\end{figure}

In Fig.~\ref{figure3}, the rate of the JP multicell transmission is overlapped with the "no-interference upper bound", while the rate of the max-Gain scheduler is overlapped with the rate of the max-SINR scheduler. Moreover, the rate difference does not decrease as the number of UEs increases and converges to a constant. This can be put in relation with Fig.~\ref{figure4}, where the interference after max-SINR scheduler decrease only slightly compared to the interference from the max-Gain scheduler, and then stop decreasing.
\section{Conclusion}
We have analyzed the asymptotic sum rate in terms of the number of UEs per cell for two channel models: The symmetric model when the path loss is the same for all the UEs and the asymetric model when the UEs are uniformly distributed in the cell. We have shown that the asymptotic properties of the interference depend strongly on the channel model. Indeed, the average interference converge to zero in the symmetric case and to a positive constant in the asymetric case. By quantifying the rate of convergence for the symmetric model, we make the observation that the interference power will remain significant at practical number of UEs. This underlines the need for other methods to manage the interference. Finally, introducing and studying intermediate less extreme models in terms of fairness is an interesting direction of research. Furthermore, the extension to multiple-antennas systems is a challenging problem to be considered in future works.

\bibliographystyle{IEEEtran}
\bibliography{Literatur_scheduling}
\end{document}